\documentclass[english,11pt]{article}
\usepackage{fullpage}
\usepackage[ruled,vlined]{algorithm2e}
\usepackage[final]{graphics,graphicx}
\usepackage{fixme}
\usepackage{amsmath}
\usepackage{amssymb}
\usepackage{babel}
\usepackage{graphics}
\usepackage{graphicx}
\usepackage{subfigure}
\usepackage{scalefnt}
\usepackage{environ}

\usepackage{tikz}
\usetikzlibrary{arrows,automata,shapes,decorations,calc,matrix}
\newtheorem{theorem}{Theorem}

\newtheorem{lemma}[theorem]{Lemma}
\newtheorem{corollary}[theorem]{Corollary}

\newenvironment{proof}{\noindent {\bf Proof }}{\qed}
\newtheorem{definition}[theorem]{Definition}

\newcommand{\qed}{\penalty 1000 \hfill\penalty 1000$\Box$\par\medskip}

\newcommand{\SolveSSG}{{\tt SolveSSG}}
\newcommand{\ModifiedValueIteration}{{\tt ModifiedValueIteration}}
\newcommand{\SolveDGG}{{\tt SolveDGG}}
\newcommand{\KwekMehlhorn}{{\tt KwekMehlhorn}}
\newcommand{\val}{{\rm val}}

\makeatletter

\begin{document}

\title{Solving simple stochastic games with few coin toss positions\thanks{The authors acknowledge support from
the Danish National Research Foundation and The National Science Foundation of
China (under the grant 61061130540) for the Sino-Danish Center for the Theory of
Interactive Computation, within which this work was performed. The authors also
acknowledge support from the Center for Research in Foundations of Electronic Markets (CFEM), supported by the Danish Strategic Research Council.}}
\author{Rasmus Ibsen-Jensen \\ Department of Computer Scinece \\ Aarhus University \and Peter Bro Miltersen \\ Department of Computer Scinece \\ Aarhus University}

\maketitle
\begin{abstract}
Gimbert and Horn gave an algorithm for solving simple stochastic
games with running time $O(r! n)$ where $n$ is the number of positions of the 
simple stochastic
game and $r$ is the number of its coin toss positions. 
Chatterjee
{\em et al.} 
pointed out that a variant of strategy iteration can be
implemented to solve this problem
in time $4^r r^{O(1)} n^{O(1)}$. In this paper, we show that an algorithm combining
value iteration with retrograde analysis 
achieves a time bound of $O(r 2^r (r \log r + n))$, thus improving both
time bounds. While the algorithm is simple, the analysis leading to
this time bound is involved, using techniques of
extremal combinatorics to identify worst case instances for the algorithm. 
\end{abstract}

\section{Introduction}
{\em Simple stochastic games} is a class of two-player zero-sum
games played on graphs that was introduced 
to the algorithms and complexity community by Condon \cite{Condon92}.
A simple stochastic game is given by a directed finite (multi-)graph 
$G = (V,E)$,
with the set of vertices $V$ also called {\em positions} and the
set of arcs $E$ also called {\em actions}. There is 
a partition of the positions into $V_1$ 
(positions belonging to player
  Max), $V_2$ (positions belonging to player Min), $V_R$ (coin toss
  positions), and a special terminal position GOAL.
Positions of $V_1, V_2, V_R$ have exactly two outgoing arcs, while the terminal
position
GOAL
has none. We shall use $r$ to denote $|V_R|$ (the number of coin toss
positions) and $n$ to denote
$|V|-1$ (the number of non-terminal positions) throughout the paper.
Between moves, a pebble is resting at one of the positions $k$.
If $k$ belongs to a player, this player should strategically
pick an outgoing arc from $k$ and move the pebble along this arc to
another node. If $k$ is a position in $V_R$, Nature picks
an outgoing arc from $k$ uniformly at random and moves the pebble
along this arc. The objective of the game for player Max is to
reach GOAL and should play so as to maximize his probability of
doing so. The objective for player Min is to minimize player Max's probability of
reaching GOAL.

A {\em strategy} for a simple stochastic game 
is a (possibly randomized) procedure
for selecting which arc or action to take, given the history of the
play so far. A {\em positional strategy} is the very special case 
of this where the choice is
deterministic and only depends on the current position, i.e.,
a positional strategy is simply a map from positions 
to actions.
If player Max plays using strategy $x$ and player Min plays using strategy
$y$, and the play starts in position $k$, a random play $p(x,y,k)$ 
of the game is induced.
We let $u(x,y)$ denote the probability that player Max will reach GOAL
in this random play.
A strategy $x^*$ for player Max is said to be {\em optimal} if for all
positions $k$ it holds that
\begin{equation}\label{opt1}
\inf_{y \in S_2} u^k(x^*,y) \geq \sup_{x \in S_1} \inf_{y \in S_2}
u^k(x, y),
\end{equation} 
where $S_1$ $(S_2)$ is the set of strategies for player Max (Min).
Similarly, a strategy $y^*$ for player Min is said to be optimal if
\begin{equation}\label{opt2}
 \sup_{x \in S_1} u^k(x,y^*) \leq \inf_{y \in S_2} \sup_{x \in S_1}
u^k(x, y).
\end{equation}
A general theorem of Liggett and Lippman 
(\cite{LL}, fixing a bug of a proof of Gillette \cite{Gil}) restricted
to simple stochastic games, implies that:
\begin{itemize}
\item{}Optimal positional
strategies $x^*, y^*$ for both players exist.
\item{}For such optimal $x^*,y^*$ and for all positions $k$, \[ \min_{y \in S_2} u^k(x^*,y) =  \max_{x \in S_1} u^k(x,y^*).\] This
number is called the {\em value} of position $k$. We shall denote
it $\val (G)_k$ and the vectors of values $\val (G)$.
\end{itemize}
In this paper, we consider {\em quantitatively solving} simple stochastic games, 
by which we mean
computing the values of all positions of the game, given an 
explicit representation of $G$. Once a simple stochastic game has been quantitatively solved, optimal strategies for both players can be found in linear time \cite{AnMi09}.
However, it was pointed out by Anne Condon twenty years ago that 
no worst case polynomial time algorithm for quantitatively solving simple stochastic games is known. By now,
finding such an algorithm is a celebrated open problem. 
Gimbert and Horn \cite{GH} pointed out
that the problem of solving simple
stochastic games parametrized by $r = |V_R|$ is 
{\em fixed parameter tractable}. That is, simple stochastic games
with ``few'' coin toss positions can be solved efficiently. The
algorithm of Gimbert and Horn runs in time $r ! n^{O(1)}$.
The next natural step in this direction is to try to find an
algorithm with a better dependence on the parameter $r$.
Thus, Dai and Ge \cite{OtherISAAC} gave a {\em randomized} algorithm with {\em expected}
running time $\sqrt{r !} n^{O(1)}$. 
Chatterjee {\em et al.} \cite{C09} pointed out that a variant of the standard
algorithm of {\em strategy iteration} devised earlier by the same
authors \cite{ChatQest} can be applied to find a solution in
time $4^r r^{O(1)} n^{O(1)}$ (they only state a time bound of
$2^{O(r)}n^{O(1)}$, but a slightly more careful analysis yields the
stated bound). The dependence on $n$ in this bound is at least quadratic.
The main result of this
paper is an algorithm running in time $O(r 2^r (r \log r + n))$, thus
improving all of the above bounds.  More precisely, we show:
\begin{theorem}\label{main}
Assuming unit cost arithmetic
on numbers of bit length up to $\Theta(r)$,
simple stochastic games with $n$ positions out of which $r$ are coin toss positions, can be quantitatively solved in
time $O(r 2^r (r \log r + n))$.
\end{theorem}
The algorithm is based on combining a variant
of {\em value iteration} \cite{Shapley,Condon93} with {\em retrograde analysis} \cite{Bellman,DGG}.
We should emphasize that the time bound of Theorem \ref{main} 
is valid {\em only} for simple stochastic
games as originally defined by Condon. 
The algorithm of Gimbert and Horn (and also the algorithm of Dai and Ge, though this is not stated in their paper) actually applies 
to a generalized version of
simple stochastic games where coin toss positions are replaced with {\em chance positions} that are allowed arbitrary
out-degree and where a not-necessarily-uniform distribution is
associated to the outgoing arcs. The complexity
of their algorithm for this more general case is 
$O(r!(|E| + p))$, where $p$ is the maximum bit-length of a
transition probability (they only claim $O(r!(n|E|+p))$, but by
using retrograde analysis in their Proposition 1, the time is reduced by a factor of $n$). The algorithm of Dai and Ge has analogous expected complexity, with the $r!$ factor replaced with $\sqrt{r!}$. 
While our algorithm and the strategy improvement algorithm of Chatterjee {\em et al.} can be generalized to also work for these
generalized simple stochastic games, the dependence on the
parameter $p$ would be much worse - in fact 
exponential in $p$. It is an interesting open problem to get
an algorithm with a complexity polynomial in  $2^r$ as well as
$p$, thereby combining the desirable features of the
algorithms based on strategy iteration and value iteration 
with the features of the algorithm of
Gimbert and Horn.


\subsection{Organization of paper}
In Section \ref{sec-alg} we present the algorithm and show how the key
to its analysis is to give upper bounds on the difference 
between the value of a given simple stochastic game and the value of a {\em time
  bounded version} of the same game. In Section \ref{sec-time}, we
then prove such upper bounds. In fact, we offer two such upper bounds: One
bound with a relatively direct proof, leading to a variant of our
algorithm with time complexity $O(r^2 2^r (r + n \log n))$ and an
{\em optimal} bound on the difference in value, shown using techniques
from extremal combinatorics, leading to an algorithm with time
complexity $O(r 2^r (r + n \log n))$. In the Conclusion section, we
briefly sketch how our technique also yields an improved upper bound
on the time complexity of the strategy iteration algorithm of
Chatterjee {\em et al.} 

\section{The algorithm}
\label{sec-alg}

\subsection{Description of the algorithm}

Our algorithm for solving simple stochastic games with few coin toss positions is the algorithm of Figure \ref{alg-main}.
\begin{figure}
\begin{center}
\begin{function}[H]
$v \leftarrow (1,0,...,0)$\;
\For{$i \in \{1,2,\ldots,2 (\ln 2^{5\max(r,6)+1})  \cdot 2^{\max(r,6)} \}$}{
$v\leftarrow \mbox{\rm $\SolveDGG$}(G,v)$\;
$v' \leftarrow v$\;
$v_{k} \leftarrow (v'_j+v'_\ell)/2,\:\mbox{\rm for all\ } k\in V_{R}$, $v_j$ and $v_\ell$ being the two successors of $v_k$\;
Round each value $v_{k}$ {\em down} to $7 r$ binary digits\;
}
$v\leftarrow \mbox{\rm $\SolveDGG$}(G,v)$\;
$v \leftarrow \mbox{\rm $\KwekMehlhorn$}(v,4^{r}  )$\;
\Return{$v$}
\caption{SolveSSG($G$)}
\end{function}

\end{center}
\vspace*{-0.5cm}
\caption{\label{alg-main}Algorithm for solving simple stochastic games}
\vspace*{-0.3cm}
\end{figure}

In this algorithm, the vectors $v$ and $v'$ are real-valued vectors
indexed by the positions of $G$. We assume the GOAL position has
the index $0$, so $v=(1,0,...,0)$ is the vector
that assigns $1$ to the GOAL position and $0$ to all other positions.  
$\SolveDGG$ is the retrograde analysis based algorithm from Proposition 1 in 
Andersson {\em et al.} \cite{DGG}
for solving {\em deterministic graphical games}. Deterministic graphical
games are defined in a similar way as simple stochastic games, but they do not
have coin toss positions, and arbitrary real payoffs are allowed at terminals.
The notation $\SolveDGG(G, v')$ means
solving the 
deterministic graphical game obtained by replacing
each coin toss position $k$  of $G$ with a terminal with payoff $v'_k$,
and returning the value vector of this deterministic
graphical game. Finally, $\KwekMehlhorn$ is the algorithm of
Kwek and Mehlhorn \cite{kwek}. $\KwekMehlhorn(v, q)$ returns a vector where
each entry $v_i$
in the vector $v$ is replaced with the smallest fraction $a/b$ with
$a/b \geq v_i$ and $b \leq q$.

The complexity analysis of the algorithm is straightforward,
given the analyses of the procedures $\SolveDGG$ and $\KwekMehlhorn$ from \cite{DGG,kwek}.
There are $O(r 2^r)$ iterations, each requiring
time $O(r \log r + n)$ for solving the deterministic graphical game. 
Finally, the Kwek-Mehlhorn algorithm requires
time $O(r)$ for each replacement, and there are only $r$ replacements to be
made, as there are only $r$ different entries different from $1,0$ 
in the vector $v$,
corresponding to the $r$ coin toss positions, by standard properties
of deterministic graphical games \cite{DGG}.

\subsection{Proof of correctness of the algorithm}\label{sec-cor}
We analyse our main algorithm by first analysing properties of a simpler
non-terminating algorithm, depicted in Figure \ref{alg2}. We shall refer
to this algorithm as {\em modified value iteration}.

Let $v^t$ be the content of the vector $v$ immediately
after executing $\SolveDGG$ in the $(t+1)$'st iteration of the loop of
$\ModifiedValueIteration$ on input $G$. To understand this variant of value iteration, we may observe that the $v^t$ vectors can be given a ``semantics'' in terms of the value of a time bounded game. 

\begin{figure}
\begin{center}
\begin{procedure}[H]
$v \leftarrow (1,0,...,0)$\;
\While{$true$}{
$v\leftarrow \mbox{\rm $\SolveDGG$}(G,v)$\;
$v' \leftarrow v$\;
$v_{k} \leftarrow (v'_j+v'_\ell)/2,\:\mbox{\rm for all\ } k\in V_{R}$, $v_j$ and $v_\ell$ being the two successors of $v_k$\;
}
\caption{ModifiedValueIteration($G$)}
\end{procedure}
\end{center}
\vspace*{-0.5cm}
\caption{\label{alg2}Modified value iteration}
\vspace*{-0.3cm}
\end{figure}

\begin{definition}
Consider the ``timed modification'' $G^t$ of the game $G$ defined
as follows. The game is played as $G$, except that play
stops and player Max loses when the play has encountered $t+1$ (not
necessarily distinct) coin toss
positions. We let $\val(G^t)_k$ be the value of $G^t$ when play starts in position $k$.
\end{definition}

\begin{lemma}\label{lem-sem}
$\forall k,t: v^t_k=\val (G^t)_k.$
\end{lemma}
\begin{proof}
Straightforward induction in $t$ (``backwards induction'').
\end{proof}

From the semantics offered by Lemma \ref{lem-sem} we immediately have
$\forall k ,t:\val (G^t)_k \leq \val (G^{t+1})_k$. Futhermore, it is true that 
$\lim_{t \rightarrow \infty} \val (G^t) = \val (G)$, where $\val (G)$ is the
value vector of $G$. This latter statement is very intuitive, given Lemma
\ref{lem-sem}, but might not be completely obvious. It may be
established rigorously as follows: 

\begin{definition}\label{def-unmodified value iteration}
For a given game $G$, let the game $\bar G^t$ be the following. The game is played as $G$, except that play stops and player Max loses when the play has encountered $t+1$ (not
necessarily distinct) {\em positions}. We let $\val(\bar G^t)_k$ be the value of $\bar
G^t$ when play starts in position $k$. 
\end{definition}
(We note that $\val(\bar G^t)$ is the valuation computed
after $t$ iterations of {\em unmodified} value iteration \cite{Condon93}.)
A very general theorem of
Mertens and Neyman \cite{MertensNeyman} linking the value of an
infinite game to the values of its time
limited versions implies that $\lim_{t \rightarrow \infty}
\val (\bar G^t) = \val (G)$. 
Also, we immediately see that for any $k$, 
$\val (\bar G^t)_k \leq \val (G^t)_k \leq \val (G)_k$, so we also have
$\lim_{t \rightarrow \infty} \val (G^t)_k = \val (G)_k$.

To relate $\SolveDGG$ of Figure \ref{alg-main} to modified value iteration of Figure \ref{alg2}, it turns out that we want to upper bound the smallest $t$ for which \[\forall i: \val(G)_i-\val(G^t)_i\leq 2^{-5r}.\] Let $T(G)$ be that $t$.  We will bound $T(G)$ using two
different approaches. The first, in Subsection \ref{sec: direct}, is
rather direct and is included to show what may be obtained using 
completely elementary means. It shows that $T(G)\leq 5(\ln 2)\cdot r^2 \cdot
2^r$, for any game $G$ with $r$ coin toss positions (Lemma~\ref{lem:direct time}). 

The second, in Subsection \ref{sec:
  extremal combinatorics}, identifies an {\em extremal} game (with
respect to convergence rate) with a given number of positions and coin
toss positions.  More precisely:

\begin{definition} \label{def-ext} Let $S_{n,r}$ be the set of simple
  stochastic games with $n$ positions out of which $r$ are coin toss
  positions. Let $G\in S_{n,r}$ be given. We say that $G$ is $t$-{\em
  extremal} if \[\max_i (\val(G)_i -\val(G^t)_i)=\max_{H\in S_{n,r}}\max_i (\val(H)_i-\val({H}^t)_i).\] We say that $G$ is {\em extremal} if it is $t$-extremal for all $t$.
\end{definition}

It is clear that $t$-extremal games exists for any choice of $n,r$ and
$t$. (That extremal games exists for any choice of $n$ and $r$ is shown
later in the present paper.) To find an extremal game, we use techniques from extremal combinatorics. By inspection of this game, we then get a better upper
bound on the convergence rate than that offered by the first
approach. We show using this approach that $T(G)\leq 2 (\ln 2^{5r+1})  \cdot 2^r$, for any game $G\in S_{n,r}$ (Corollary \ref{cor:howfast}).

Assuming that an upper bound on $T(G)$ is available, we are now ready to finish the proof of correctness of the main
algorithm. We will only do so explicitly for the bound on $T(G)$
obtained by the second approach from Subsection \ref{sec: extremal
  combinatorics} (the weaker bound implies correctness of a version of
the algorithm performing more iterations of its main loop).
From Corollary \ref{cor:howfast}, we have that for any game $G\in S_{n,r}$, $\val (G^T)_k$ and hence, modified value iteration, approximates $\val (G)_k$ within an additive error of $2^{-5r}$ for $t\geq 2 (\ln 2^{5r+1})  \cdot 2^r$ and $k$ being any position. 
$\SolveSSG$ differs from $\ModifiedValueIteration$ by rounding down
the values in the vector $v$ in each iteration. 
Let $\tilde v^t$ be the content of the vector $v$ immediately
after executing $\SolveDGG$ in the $t$'th iteration of the loop of
$\SolveSSG$.
We want to compare $\val (G^t)_k$ with $\tilde v^t_k$ for any $k$. 
As each number
is rounded down by less than $2^{-7r}$ in each iteration of the loop and recalling Lemma \ref{lem-sem},
we see by induction that 
\[\val (G^t)_k - t2^{-7r} \leq \tilde v^t_k \leq \val (G^t)_k.\]
In particular, when $t = 2 (\ln 2^{5r+1})  \cdot 2^r$, we have
that $\tilde v^t_k$ approximates $\val (G)_k $ within $2^{-5r} +
2 (\ln 2^{5r+1})  \cdot 2^r 2^{-7r} < 2^{-4r}$, for any $k$, as we can assume $r \geq 6$ by the code of $\SolveSSG$ of Figure \ref{alg-main}.

Lemma 2 of Condon \cite{Condon92} states that the value of a position
in a simple stochastic game with $n$ non-terminal positions can be
written as a fraction with integral numerator and 
denominator at most $4^n$. As pointed out by 
Chatterjee {\em et al.} \cite{C09}, it is straightforward to see that
her proof in fact gives an upper bound of $4^r$, where $r$
is the number of coin toss positions. It is well-known that two distinct fractions
with denominator at most $m \geq 2$ differ by at least $\frac{1}{m(m-1)}$.
Therefore, since $\tilde v^t_k$
  approximates $\val (G)_k$ within $2^{-4r}< \frac{1}{4^r\cdot (4^r-1)}$ from below,
we in fact have that $\val (G)_k$ is the smallest
rational number $p/q$ so that $q \leq 4^r$ and $p/q \geq \tilde v^t_k$. 
Therefore, the Kwek-Mehlhorn algorithm applied to
$\tilde v^t_k$ correctly computes $\val (G)_k$, and we are done.

We can not use the bound on $T(G)$ obtained by the first direct approach (in
Subsection \ref{sec: direct}) to show the correctness of $\SolveSSG$, but we can show the correctness of the version of
it that runs the main loop an additional factor of $O(r)$ times, that is,
$i$ should range over $\{1,2,\dots,5(\ln 2) \cdot r^2 2^{r}\}$ instead of over $\{1,2,\dots,2 (\ln 2^{5r+1})  \cdot 2^r\}$.

%
%


\section{Bounds on the convergence rate\label{sec-time}}

\subsection{A direct approach\label{sec: direct}}

\begin{lemma}\label{lem:ssg proof}Let $G\in S_{n,r}$ be given. For all positions 
$k$ and all integers $i \geq 1$, we have \[\val (G)_k-\val (G^{i\cdot r})_k 
\leq 
(1-2^{-r})^i.\]
\end{lemma}
\begin{proof}
If $\val (G)_k = 0$, we also have $\val (G^{i\cdot r})_k = 0$, so the inequality
holds. Therefore, we can assume that $\val (G)_k > 0$.

Fix some optimal positional strategy, $x$, for Max in $G$. Let
$y$ be any pure (i.e., deterministic, but not necessarily
positional) strategy for Min with the
property that $y$ guarantees that the pebble will not reach GOAL after
having been in a position of value 0 (in particular, any best reply 
to any strategy of Max, including $x$, clearly
has this property). 

The two strategies $x$ and $y$ together induce a probability
space $\sigma_k$ on the set of plays of the game, starting in position $k$. Let the probability measure on plays of $G^t$ associated with
this strategy be denoted $\Pr_{\sigma_k}$. Let
$W_k$ be the event that this random play reaches GOAL. 
We shall also consider the event $W_k$ to be a set of plays. 
Note that any position occurring in any
play in $W_k$ has non-zero value, by definition of $y$.

{\em Claim.}
There is a play in $W_k$ where each position
occurs at most once.

{\em Proof of Claim.} Assume to the contrary that for all plays in $W_k$,
some position occurs at least twice. Let $y'$ be the modification of
$y$ where the second time a position, $v$, in $V_2$ is entered in a
given play, $y$ takes the same action as was used the first time
$v$ occurred. Let $W'$ be the set of plays generated by $x$
and $y$ for which the pebble reaches GOAL. We claim that $W'$ is in fact 
the empty
set. Indeed, if $W'$ contains any play $q$, we can obtain a play in 
$W'$ where each position occurs only once, by removing all transitions in
$q$ occurring between repetitions of the same position. Such a play 
is also an element of $W_k$, contradicting the assumption 
that all plays in $W_k$ has a position occurring twice. 
The emptiness of $W'$ shows that the strategy
$x$ does not guarantee that GOAL is reached with positive
probability, when play starts in $k$. This contradicts either
that $x$ is optimal or that $\val (G)_k > 0$. We therefore conclude that
our assumption is incorrect, and that there is a
play $q$ in $W_k$ where each position occurs only once, as desired. {\em This completes the proof of the claim.}

The
probability according to the probability measure $\sigma_k$ that a
given play where each coin toss position occurs only once occurs, 
is at least $2^{-r}$. 

Let $W_k^i$ be the set of plays in $W_k$ that contains at most $i$
occurrences of coin toss 
positions (and also let $W_k^i$ denote the corresponding event with respect to the measure $\sigma_k$). 
Since the above claim holds for any position $k$ of non-zero value
and plays in $W_k$ only visits positions of non-zero value, we
see that $\Pr_{\sigma_k}[\neg W_k^{i\cdot r}| W_k] \leq (1-2^{-r})^i$, for any
$i$. Since $x$ is optimal, we also have  $\Pr_{\sigma_k}[W_k] \geq \val (G)_k $. Therefore,
\begin{eqnarray*}
\Pr_{\sigma_k}[W_k^{i\cdot r}] & = & \Pr_{\sigma_k}[W_k] - \Pr_{\sigma_k}[\neg W_k^{i\cdot r}|W_k] \Pr_{\sigma_k}[W_k]\\
&\geq& \val (G)_k -(1-2^{-r})^i 
\end{eqnarray*}
The above derivation is true for any $y$ guaranteeing that no play can
enter a position of value 0 and then reach GOAL, and therefore it is also
true for $y$ being the optimal strategy in the time-limited game,
$G^{i\cdot r}$. In that case, we have $\Pr_{\sigma_k} [W_k^{i\cdot r}] \leq \val (G^{i\cdot r})_k$. 
We can therefore conclude that $\val (G^{i\cdot r})_k \geq \val (G)_k-(1-2^{-r})^i$, as desired.
\end{proof}

\begin{lemma}\label{lem:direct time}
Let $G\in S_{n,r}$ be given. 
\[T(G)\leq 5(\ln 2)\cdot r^2 \cdot 2^r\]
\end{lemma}
\begin{proof}
We will show that for any $t\geq 5(\ln 2)\cdot r^2 \cdot 2^r$ and any $k$, we have that $\val (G)_k -\val (G^{t})_k <2^{-5r}$. 

From Lemma \ref{lem:ssg proof} we have that $\forall i,k:
\val(G)_k-\val(G^{i\cdot r})_k <(1-2^{-r})^{i}$. Thus,
\[ \val(G)_k-\val(G^{t})_k < (1-2^{-r})^{t/r} =((1-2^{-r})^{2^r})^\frac{t}{r\cdot 2^r}<e^{-\frac{t}{r\cdot 2^r}}
\leq e^{-\frac{5 (\ln 2) r^2 2^r}{r\cdot 2^r}}
= 2^{-5r}. \]
\end{proof}

\subsection{An extremal combinatorics approach\label{sec: extremal combinatorics}}
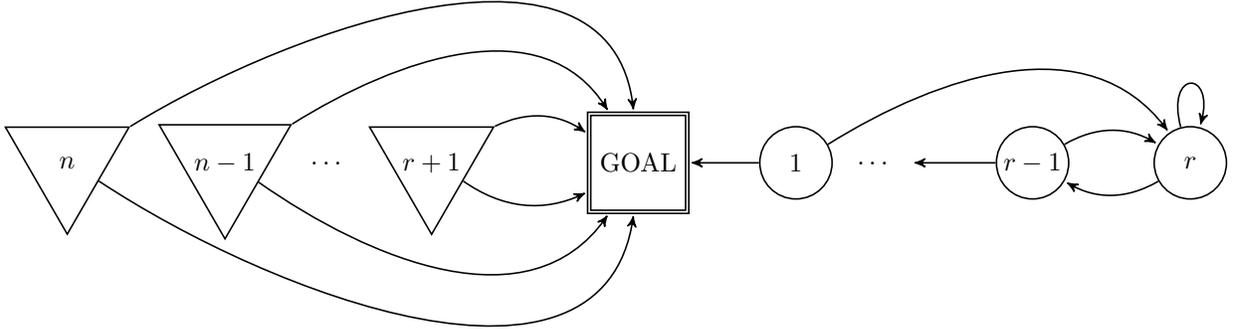
\begin{figure}
\begin{centering}
\resizebox{\textwidth}{!}{
\begin{tikzpicture}[->,>=stealth',shorten >=1pt,auto,node distance=5.5cm*0.5,
                    semithick,scale=0.2 ]
      
\tikzstyle{every state}=[fill=white,draw=black,text=black,font=\small , inner sep=-0.05cm]
\tikzstyle{max}=[state,regular polygon,regular polygon sides=3,minimum height=1.9cm,minimum width=1.9cm]
\tikzstyle{min}=[max,regular polygon rotate=180]

    \node[min] (m){$n$};
    \node[min] (m1)[right of=m,node distance=4.2cm*0.5,]{$n-1$};
    \node[state,draw=white,node distance=5.5cm*0.25] (dots1) [right of=m1] {$\dots$};
    \node[min] (m2) [right of=m1] {$r+1$};
    \node[state,accepting,regular polygon,regular polygon sides=4] (perp) [right of=m2] {GOAL};

    \node[state,node distance=4.2cm*0.5] (q1) [right of=perp] {$1$};
    \node[state,draw=white,node distance=4.2cm*0.25] (dots) [right of=q1] {$\dots$};
    \node[state,node distance=4.2cm*0.5] (q2) [right of=dots] {$r-1$};

    \node[state,node distance=4.2cm*0.5] (q3) [right of=q2] {$r$};
\path 
(m)  edge [bend right,in=-95] (perp)
(m)  edge[bend left,in=95] (perp)
(m1)  edge [bend right,in=-120] (perp)
(m1)  edge[bend left,in=120] (perp)
(m2)  edge [bend right] (perp)
(m2)  edge[bend left] (perp)

(q1) edge (perp)
(q1) edge[bend left,in=125] (q3)
(q2) edge (dots)
(q2) edge[bend left] (q3)
(q3) edge [loop above]
(q3) edge[bend left] (q2);

\end{tikzpicture}
}
\end{centering}
\vspace*{-1.3cm}
\caption{The extremal game $E_{n,r}$.
Circle nodes are coin toss positions, triangle nodes are Min positions
and the node labeled GOAL is 
the GOAL position.}
\label{fig-extreme}
\end{figure}
The game of Figure \ref{fig-extreme} is a game in $S_{n,r}$. We will refer to this game as
$E_{n,r}$. $E_{n,r}$ consists of no Max-positions. Each Min-position in $E_{n,r}$ has GOAL as a successor twice. The $i$'th coin toss position, for $i\geq 2$, has the $(i-1)$'st and the $r$'th coin toss position as successors. The first coin toss position has GOAL and the $r$'th coin toss position as successors. The game is very similar to a 
simple stochastic game used as an example by Condon \cite{Condon93}
to show that unmodified value iteration converges slowly.

In this subsection we will show that $E_{n,r}$ is an extremal game in
the sense of Definition \ref{def-ext} and upper bound $T(E_{n,r})$,
thereby upper bounding $T(G)$ for all $G\in S_{n,r}$. 

The first two lemmas in this subsection concerns assumptions about $t$-extremal games we can make without loss of generality.
\begin{lemma}\label{lem:max}For all $n,r,t$, there is a $t$-extremal game in $S_{n,r}$ with
$V_1 = \emptyset$, i.e., without containing positions belonging to
Max.
\end{lemma}
\begin{proof}
Take any $t$-extremal game $G\in S_{n,r}$. Let $x$ be an optimal positional strategy
for Max in this game. Now replace each position belonging to Max with
a position belonging to Min with both outgoing arcs making the 
choice specified by $x$. Call the resulting game $H$. We claim that
$H$ is also $t$-extremal. First, clearly, each position $k$ of $H$ has the
same value as is has in $G$, i.e., $\val (G)_k = \val (H)_k$. 
Also, if we compare the values of
the positions of the games $H^t$
and $G^t$ defined in the statement of Lemma \ref{lem-sem},
we see that $\val (H^t)_k \leq \val (G^t)_k$, since the only difference between
$H^t$ and $G^t$ is that player Max has more options in the latter game.
Therefore, $ \val (H)_k  -  \val (H^t )_k \geq  \val (G)_k - \val (G^t)_k$ so 
$H$ must also be $t$-extremal.
\end{proof}

\begin{lemma}
\label{lem:Value-1}For all $n,r,t$, there 
exists a $t$-extremal game in $S_{n,r}$, where all positions have value one and where
no positions belong to player Max.
\end{lemma}
\begin{proof}
By Lemma \ref{lem:max}, we can pick a $t$-extremal game $G$ in $S_{n,r}$ where
no positions belong to player Max.  
Suppose that not all positions in $G$ have value
1. Then, it is easy to see that the set of positions of value $0$ is non-empty. Let this
set be $N$. Let $H$ be the game where all arcs into $N$
are redirected to GOAL. Clearly, all positions in this game have
value $1$. We claim that $H$ is also $t$-extremal.

Fix a position $k$. 
We shall show that $\val (H)_k - \val (H^t)_k \geq \val (G)_k - \val (G^t)_k$ and
we shall be done.
Let $\sigma_k$ be a (not necessarily positional) optimal strategy for
player Min in $G^t$ for plays starting in $k$ 
and let the probability measure on plays of $G^t$ associated with
this strategy be denoted $\Pr_{\sigma_k}$.
As $\sigma_k$ is also a strategy that can be played in $G$,
we have $\Pr_{\sigma_k}[\mbox{\rm Play does not reach $N$}] \geq \val (G)_k $.
Also, by definition, $\Pr_{\sigma_k}[\mbox{\rm Play reaches GOAL}] = \val (G^t)_k$.
That is,
\[ \Pr_{\sigma_k}[\mbox{\rm Play reaches neither GOAL nor $N$}] \geq \val
(G)_k- \val (G^t)_k. \]
Let $\bar \sigma_k$ be an optimal strategy for plays starting in $k$ for
player Min in $H^t$.
This strategy can also be used in $G^t$. Let
the probability distribution on plays of $G^t$ associated with
this strategy be denoted $\Pr_{\bar \sigma_k}$.
Note that plays reaching GOAL in $H^t$ correspond to 
those plays reaching
either GOAL or $N$ in $G^t$. Thus,
by definition, 
$\Pr_{\bar \sigma_k}[\mbox{\rm Play reaches neither GOAL nor $N$}] = 1 - \val(H^t)_k$. As $\sigma_k$ can be used in $H^t$ where $\bar \sigma_k$ is
optimal, we have \[1 - \val(H^t)_k \geq \val (G)_k- \val (G^t)_k.\] But since
$\val(H)_k  = 1$, this is the desired inequality 
\[\val(H)_k - \val(H^t)_k \geq \val (G)_k- \val (G^t)_k.\]
\end{proof}


The next lemma will be used to derive a ordering of the positions in any game $G$ satisfying the restrictions of Lemma \ref{lem:Value-1}.

\begin{lemma}\label{Hlemma}
Let $G$ be a game without Max positions in which all positions 
have value one. Let $V'$ be a non-empty set of positions of $G$ that
does not include GOAL. Then, at least one of the following two cases hold:
\begin{enumerate}
\item{} $V'$ contains a Min position
with both successors outside of $V'$ or
\item{} $V'$ contains a coin toss position with at least one
successor outside of $V'$.
\end{enumerate}
\end{lemma}
\begin{proof}
Suppose not. Then the Min-player can force play to stay within $V'$ when
play starts in $V'$.
Thus, the values of all positions in $V'$ are $0$,
a contradiction.
\end{proof}

The following lemma will be used several times
to change the structure of a game while only making it more extremal,
eventually making the game into the specific game $E_{n,r}$
(in the context of extremal combinatorics, this is a standard technique
pioneered by Moon and
Moser \cite{MoonMoser}).
\begin{lemma}
\label{lem:change-child}
Given a game $G$. Let $c$ be a coin toss position in $G$
and let $k$ be an immediate 
successor position $k$ of $c$. Also, let a position $k'$ 
with the following property be given:
$\forall t: \val (G^t)_{k'} \leq \val (G^t)_{k}$. 
Let $H$ be the game where the arc from $c$ to $k$ is
redirected to $k'$. Then, $\forall t,j:  \val (H^t)_{j} \leq \val (G^t)_{j}$.
\end{lemma}
\begin{proof}
In this proof we will throughout refer to the
properties of $\ModifiedValueIteration$ and use Lemma \ref{lem-sem}.
We show by induction in $t$ that
$\forall j,t: \val (H^t)_{j} \leq \val (G^t)_{j}$.
For $t=0$ we have $\val (H^t)_{j} = \val (G^t)_{j}$ by inspection of the algorithm.
Now assume that the inequality holds for all values smaller than $t$ and for all
positions $i$ and we will show that it holds for $t$ and all positions $j$.
Consider a fixed position $j$. 
There are three cases.
\enumerate{}
\item{}
{\em The position $j$ belongs to Max or Min}.
In this case, we observe that the function computed by $\SolveDGG$  to determine the value of position
$j$ in $\ModifiedValueIteration$ is a monotonously increasing
function. Also, the deterministic graphical game obtained when replacing coin toss
positions with terminals is the same for $G$ and for $H$.
By the induction hypothesis, we have that 
for all $i$, $\val (H^{t-1})_i \leq \val (G^{t-1})_i $. So, 
$\val (H^{t})_j \leq \val (G^{t})_j $.
\item{}{\em The position $j$ is a coin toss position, but not $c$}.
In this case, we
have 
\[\val (G^t)_j = \frac{1}{2} \val (G^{t-1})_a + \frac{1}{2} \val (G^{t-1})_b,\] 
and
\[\val (H^t)_j = \frac{1}{2} \val (H^{t-1})_a + \frac{1}{2} \val (H^{t-1})_b\]
where $a$ and $b$ are the successors of $j$. By the induction
hypothesis, $\val(H^{t-1})_a \leq \val(G^{t-1})_a$ and  
$\val(H^{t-1})_b \leq \val(G^{t-1})_b$.
Again, we have $\val(H^{t})_j \leq \val(G^{t})_j$.
\item{}{\em The position $j$ is equal to $c$.}
In this case, we have
\[\val(G^t)_c  = \frac{1}{2} \val(G^{t-1})_a + \frac{1}{2} \val(G^{t-1})_k\]
where $a$ and $k$ are the successors of $c$ in $G$ while
\[\val(H^t)_c = \frac{1}{2} \val(H^{t-1})_a + \frac{1}{2} \val(H^{t-1})_{k'}.\] By the induction hypothesis we have
$\val(H^{t-1})_{a} \leq \val(G^{t-1})_{a}$. We also have that $\val(H^{t-1})_{k'} \leq \val(G^{t-1})_{k'}$ which is, by assumption, at
most $\val(G^{t-1})_{k}$. So, we have $\val(H^t)_c \leq \val(G^t)_c$.
\end{proof}

\begin{theorem}\label{thm-extreeemal}
$E_{n,r}$ is an extremal game in $S_{n,r}$.
\end{theorem}

\begin{proof}
Let $H_0\in S_{n,r}$, where $V_1=\emptyset$ and all positions in $H_0$
have value one. We will show that for all $t$ we have that $\max_k (\val (E_{n,r})_k - \val (E_{n,r}^t)_k) \geq \max_{k'}(\val (H_0)_{k'} - \val (H_0^t)_{k'})$. Since by Lemma \ref{lem:Value-1}, we can take $H_0$ to be a $t$-extremal game for any $t$, $E_{n,r}$ is a $t$-extremal game for all $t$ and is hence an extremal game. 

To illustrate the proof we will use as running example the game in Figure \ref{fig:w0}.

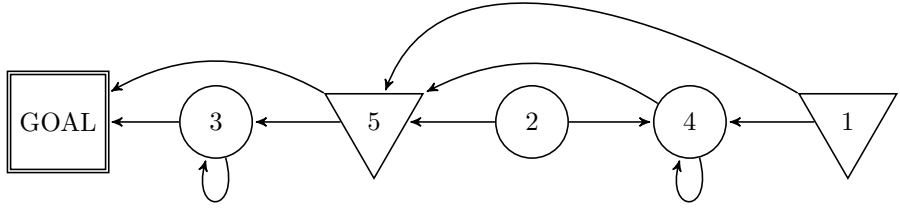
\begin{figure}
\centering
\begin{tikzpicture}[scale=0.2,->,>=stealth',shorten >=1pt,auto,node distance=4.2cm*0.5,
                    semithick]
      
\tikzstyle{every state}=[fill=white,draw=black,text=black,font=\small , inner sep=-0.05cm]
\tikzstyle{max}=[state,regular polygon,regular polygon sides=3,inner sep=0cm,minimum height=1.5cm,minimum width=1.5cm]
\tikzstyle{min}=[max,regular polygon rotate=180]

    \node[state,accepting,regular polygon,regular polygon sides=4,inner sep=-0.05cm] (perp)  {GOAL};
    \node[state] (q1) [right of=perp] {3};
    \node[min] (q2) [right of=q1] {5};
    \node[state] (q3) [right of=q2] {2};
    \node[state] (q4) [right of=q3] {4};
    \node[min] (q5) [right of=q4] {1};

\path 
(q1)  edge (perp)
(q1)  edge [loop below] (q1)
(q2)  edge [bend right] (perp)
(q2)  edge (q1)
(q3)  edge  (q2)
(q3)  edge  (q4)
(q4)  edge  [loop below]
(q4)  edge[bend right] (q2)
(q5)  edge  (q4)
(q5)  edge [bend right,in=-110] (q2);

\end{tikzpicture}
\caption{An example of $H_0$, with $n=5$ and $r=3$.
Circle nodes are coin toss positions, triangle nodes are Min positions
and the node labeled GOAL is 
the GOAL position. Note that this particular $H_0$ is not extremal.}
\label{fig:w0}
\end{figure}

We shall construct a sequence
$H_1, H_2, \ldots, H_n$ of games (which will in fact
be identical to $H_0$, except that positions have been renumbered), 
so that $H_k$ has the following property $P_k$.

{\bf Property} $P_k$: 
Any Min position $j$ among the positions
$1,2,\ldots,k$ has {\em all} successors within the set
$\{1,2,\ldots,j-1,\mbox{GOAL}\}$. Any coin toss position $j$ among
the positions $\{1,2,\ldots,k\}$ has {\em at least one} successor
within $\{1,2,\ldots,j-1,\mbox{GOAL}\}$.

Suppose we already constructed $H_j$ for $j < k$. 
We show how to construct $H_k$ based on $H_{k-1}$. Applying Lemma 
\ref{Hlemma} to the game $H_{k-1}$
with $V' = \{k, \ldots, n\}$, we find among the positions $k,\ldots,n$
either a coin toss position $u$ with one successor in $\{1,2,\ldots,k-1,\mbox{GOAL}\}$
or a Min-position $u$ with all successors in $\{1,2,\ldots,k-1,\mbox{GOAL}\}$. In either case,
we renumber $u$ to $k$ and $k$ to $u$ and let the resulting game be $H_k$.

Figure \ref{fig:wn} shows $H_n$ for the case of our running example.
\begin{figure}
\centering
\begin{tikzpicture}[scale=0.2,->,>=stealth',shorten >=1pt,auto,node distance=4.2cm*0.5,
                    semithick , inner sep=-0.05cm]
      
\tikzstyle{every state}=[fill=white,draw=black,text=black,font=\small]
\tikzstyle{max}=[state,regular polygon,regular polygon sides=3,inner sep=0cm,minimum height=1.5cm,minimum width=1.5cm]
\tikzstyle{min}=[max,regular polygon rotate=180]

    \node[state,accepting,regular polygon,regular polygon sides=4,inner sep=-0.05cm] (perp)  {GOAL};
    \node[state] (q1) [right of=perp] {1};
    \node[min] (q2) [right of=q1] {2};
    \node[state] (q3) [right of=q2] {3};
    \node[state] (q4) [right of=q3] {4};
    \node[min] (q5) [right of=q4] {5};

\path 
(q1)  edge (perp)
(q1)  edge [loop below] (q1)
(q2)  edge [bend right] (perp)
(q2)  edge (q1)
(q3)  edge  (q2)
(q3)  edge  (q4)
(q4)  edge  [loop below]
(q4)  edge[bend right] (q2)
(q5)  edge  (q4)
(q5)  edge [bend right,in=-110] (q2);

\end{tikzpicture}
\caption{$H_n$, if $H_0$ is the game in Figure \ref{fig:w0}. }
\label{fig:wn}
\end{figure}
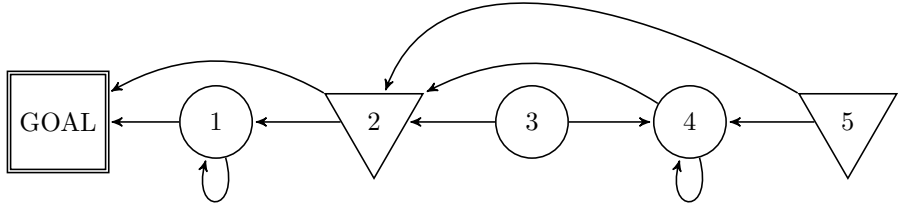

Each coin toss position in $H_n$ has at least one successor with a lower index than
itself. (Recall that GOAL has index 0.) 
In the following, we call this successor the {\em lower} successor and the
other successor the {\em higher} successor. If both successors in
fact have a lower index than the position, we choose one of the
two successor arbitrarily as the higher successor. 

We now make a series of transformations from $H_n$ generating
a new sequence of games. This
will take us outside the set of games in $S_{n,r}$, but the last
game in the sequence will again be in $S_{n,r}$. For each transformation, from $G'$ to $G''$, we will show that $\val({G''}^t)_k\leq \val({G'}^t)_k$ for all $t$ and $k$. For the final game $E_{n,r}$ we arrive at, we clearly have that $\val(E_{n,r})_{k}=1$ for all $k$. This is in fact also true for all intermediate games, but we shall not need that fact.

For each of the original non-terminal positions $1,2,\ldots,n$ in $H_n$, we add a Min-position. We
assign index $n+j$ to the Min-position associated with position $j$. We let the two successors of position $n+j$ be
$j$ and $n+j-1$,
except for the case of $n+1$, where we let the two successors be $1$ and GOAL. Let the resulting game
be denoted $H'$. For our running example, $H'$ is shown in Figure \ref{fig:w'}. 

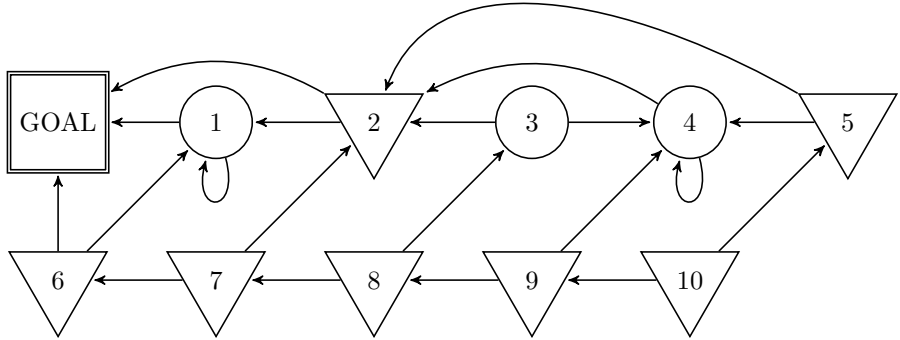
\begin{figure}
\centering
\begin{tikzpicture}[scale=0.2,->,>=stealth',shorten >=1pt,auto,node distance=4.2cm*0.5,
                    semithick]
      
\tikzstyle{every state}=[fill=white,draw=black,text=black,font=\small]
\tikzstyle{max}=[state,regular polygon,regular polygon sides=3,inner sep=0cm,minimum height=1.5cm,minimum width=1.5cm]
\tikzstyle{min}=[max,regular polygon rotate=180]

    \node[state,accepting,regular polygon,regular polygon sides=4,inner sep=-0.05cm] (perp)  {GOAL};
    \node[state] (q1) [right of=perp] {1};
    \node[min] (q2) [right of=q1] {2};
    \node[state] (q3) [right of=q2] {3};
    \node[state] (q4) [right of=q3] {4};
    \node[min] (q5) [right of=q4] {5};
    \node[min] (m1) [below of=perp] {6};
    \node[min] (m2) [below of=q1] {7};
    \node[min] (m3) [below of=q2]{8};
    \node[min] (m4) [below of=q3]{9};
    \node[min] (m5) [below of=q4]{10};

\path 
(q1)  edge (perp)
(q1)  edge [loop below] (q1)
(q2)  edge [bend right] (perp)
(q2)  edge (q1)
(q3)  edge  (q2)
(q3)  edge  (q4)
(q4)  edge  [loop below]
(q4)  edge[bend right] (q2)
(q5)  edge  (q4)
(q5)  edge [bend right,in=-110] (q2)
(m1) edge (q1)
(m1) edge (perp)
(m2) edge (q2)
(m2) edge (m1)
(m3) edge (q3)
(m3) edge (m2)
(m4) edge (q4)
(m4) edge (m3)
(m5) edge  (q5)
(m5) edge (m4);

\end{tikzpicture}
\caption{$H'$, if $H_0$ is the game in Figure \ref{fig:w0}. }
\label{fig:w'}
\end{figure}

Applying Lemma \ref{lem-sem} and inspecting the code of $\ModifiedValueIteration$, we have that \[\forall k\in \{1,2,\ldots,n\},t: \val({H'}^t)_k= \val(H_n^t)_k.\] We will only use that $\forall k\in \{1,2,\ldots,n\},t: \val({H'}^t)_k\leq  \val(H_n^t)_k$.

We also have the following fact: $\val ({H'}^t)_{2n} =\min_{j \in \{1,2,\ldots,n\}} \val ({H'}^t)_{j}$. Indeed, we can argue by induction in $k$ that
$\val ({H'}^t)_{n+k}=\min_{j \in \{1,2,\ldots,k\}} \val ({H'}^t)_{j}$. For the base case,
we have that $\val ({H'}^t)_{n+1} = \min(1, \val ({H'}^t)_{1})$. For $j > 1$, we have \[\val ({H'}^t)_{n+j} = \min(\val ({H'}^t)_{j}, \val ({H'}^t)_{n+j-1}),\]
completing the proof.
This property, and the fact that the proof only used information about the successors of $n+k$ for $k\in \{1,2,\ldots,n\}$, allows us to apply Lemma \ref{lem:change-child} iteratively 
to modify the game by changing the higher successor of each and 
every coin toss position to
be the position $2n$. We denote this game
$H'_0$. 

For our running example, $H'_0$ is shown in Figure \ref{fig:w'0}. 

\begin{figure}
\centering
\begin{tikzpicture}[scale=0.2,->,>=stealth',shorten >=1pt,auto,node distance=4.2cm*0.5,
                    semithick]
      
\tikzstyle{every state}=[fill=white,draw=black,text=black,font=\small, inner sep=-0.05cm]
\tikzstyle{max}=[state,regular polygon,regular polygon sides=3,inner sep=0cm,minimum height=1.5cm,minimum width=1.5cm]
\tikzstyle{min}=[max,regular polygon rotate=180]

    \node[state,accepting,regular polygon,regular polygon sides=4,inner sep=-0.05cm] (perp)  {GOAL};
    \node[state] (q1) [right of=perp] {1};
    \node[min] (q2) [right of=q1] {2};
    \node[state] (q3) [right of=q2] {3};
    \node[state] (q4) [right of=q3] {4};
    \node[min] (q5) [right of=q4] {5};
    \node[min] (m1) [below of=perp] {6};
    \node[min] (m2) [below of=q1] {7};
    \node[min] (m3) [below of=q2]{8};
    \node[min] (m4) [below of=q3]{9};
    \node[min] (m5) [below of=q4]{10};

\path 
(q1)  edge (perp)
(q1)  edge [out=-45,in=150]  (m5)
(q2)  edge [bend right] (perp)
(q2)  edge (q1)
(q3)  edge  (q2)
(q3)  edge (m5)
(q4)  edge[bend right] (q2)
(q4)  edge  (m5)
(q5)  edge  (q4)
(q5)  edge [bend right,in=-110] (q2)
(m1) edge (q1)
(m1) edge (perp)
(m2) edge (q2)
(m2) edge (m1)
(m3) edge (q3)
(m3) edge (m2)
(m4) edge  (q4)
(m4) edge (m3)
(m5) edge  (q5)
(m5) edge (m4);

\end{tikzpicture}
\caption{$H'_0$, if $H_0$ is the game of Figure \ref{fig:w0}. }
\label{fig:w'0}
\end{figure}

Let the coin toss positions in $H'_0$ be $i_1 < i_2 < \cdots < i_r$.
Then, we define a sequence of games $H'_1, H'_2, \ldots, H'_r$ as follows.
We define $H'_1$ from $H'_0$ by changing the lower successor of $i_1$ to GOAL.
For $j > 1$, we define $H'_j$ from $H'_{j-1}$ by changing the lower 
successor of $i_j$ to be $i_{j-1}$.

For our running example, $H'_r$ is given in Figure \ref{fig:w'r}. 

\begin{figure}
\centering
\begin{tikzpicture}[scale=0.2,->,>=stealth',shorten >=1pt,auto,node distance=4.2cm*0.5,
                    semithick]
      
\tikzstyle{every state}=[fill=white,draw=black,text=black,font=\small, inner sep=-0.05cm]
\tikzstyle{max}=[state,regular polygon,regular polygon sides=3,inner sep=0cm,minimum height=1.5cm,minimum width=1.5cm]
\tikzstyle{min}=[max,regular polygon rotate=180]

    \node[state,accepting,regular polygon,regular polygon sides=4,inner sep=-0.05cm] (perp)  {GOAL};
    \node[state] (q1) [right of=perp] {1};
    \node[min] (q2) [right of=q1] {2};
    \node[state] (q3) [right of=q2] {3};
    \node[state] (q4) [right of=q3] {4};
    \node[min] (q5) [right of=q4] {5};
    \node[min] (m1) [below of=perp] {6};
    \node[min] (m2) [below of=q1] {7};
    \node[min] (m3) [below of=q2]{8};
    \node[min] (m4) [below of=q3]{9};
    \node[min] (m5) [below of=q4]{10};

\path 
(q1)  edge (perp)
(q1)  edge [out=-45,in=150]  (m5)
(q2)  edge [bend right] (perp)
(q2)  edge (q1)
(q3)  edge [bend right]  (q1)
(q3)  edge (m5)
(q4)  edge (q3)
(q4)  edge  (m5)
(q5)  edge  (q4)
(q5)  edge [bend right,in=-110] (q2)
(m1) edge (q1)
(m1) edge (perp)
(m2) edge (q2)
(m2) edge (m1)
(m3) edge (q3)
(m3) edge (m2)
(m4) edge  (q4)
(m4) edge (m3)
(m5) edge  (q5)
(m5) edge (m4);

\end{tikzpicture}
\caption{$H'_r$, if $H_0$ is the game in Figure \ref{fig:w0}. }
\label{fig:w'r}
\end{figure}

{\em Claim.}
For $t \geq 0, j \in \{1,\ldots,r+1\}$, the following holds. 
For a position with index 
$k$ strictly smaller than $i_{j}$, we have $\val({H'}^t_{j-1})_k \geq \val({H'}^t_{j-1})_{i_{j-1}}$.
Here, by convention, we let $i_0$ be the GOAL position when considering
the
statement for $j=1$ and we let $i_{r+1}$ be $\infty$ when
considering
the statement for $j=r+1$.

{\em Proof of claim.}
The proof is by induction in $j$.

Clearly, $\val ({H'}_{j'}^t)_k = 1$ for all positions $k$ in $1,2,\ldots,i_{1}-1$ and for all $j'$ and $t$,
so this settles the base case of $j=1$. 

For larger values of $j$, 
and $k < i_{j-1}$, we have by construction that $\val ({H'}_{j-1}^t)_k=\val ({H'}_{j-2}^t)_k$.
Now, $k\geq i_{j-1}$. Therefore there are two cases. Either $k$ is a coin toss position or $k$ is a Min-position. 

If $k$ is a coin toss position, we can without loss of generality assume that $k=i_{j-2}$, since it has the smallest value among all coin toss positions, by the induction hypothesis. Therefore, we only need to show that $\forall t: \val ({H'}_{j-1}^t)_{i_{j-2}}\geq \val ({H'}_{j-1}^t)_{i_{j-1}}$.
For $j=2$, we have that $\forall t: \val({H'}_{1}^t)_{i_0}=1\geq \val({H'}_{1}^t)_{i_1}$. For $j\geq 3$, 
we have by the properties of the algorithm that 
for $\val ({H'}_{j-1}^0)_{i_{j-2}}=0=\val ({H'}_{j-1}^0)_{i_{j-1}}$ and 
$\forall t\geq 1: \val ({H'}_{j-1}^t)_{i_{j-2}} =\frac{1}{2} \val ({H'}^{t-1}_{j-1})_{i_{j-3}}+\frac{1}{2}\val ({H'}_{j-1}^{t-1})_{i_{2n}}$. 
We have by the induction hypothesis that 
\begin{eqnarray*}
\frac{1}{2}\val ({H'}_{j-1}^{t-1})_{i_{j-3}}+\frac{1}{2}\val ({H'}_{j-1}^{t-1})_{i_{2n}} & \geq & \frac{1}{2}\val ({H'}_{j-1}^{t-1})_{i_{j-2}}+\frac{1}{2}\val ({H'}_{j-1}^{t-1})_{i_{2n}} \\
& = & \val ({H'}_{j-1}^{t})_{i_{j-1}}.
\end{eqnarray*}

If $k$ is a Min-position in $\{i_{j-1}+1,\ldots,i_{j}-1\}$, assume to the contrary
that some $k$ fails to satisfy $\val ({H'}_{j-1}^{t})_k \geq \val ({H'}_{j-1}^{t})_{i_{j-1}}$. Consider
the smallest such $k$. As $k$ is a Min-position with two successors 
both of which are smaller than $k$, we have that $\val ({H'}_{j-1}^{t})_k$
is the minimum of two numbers, both
 of which are at least $\val ({H'}_{j-1}^{t})_{i_{j-1}}$,
either by the induction hypothesis or the assumption that $k$ is minimal.
{\em This completes the proof of the claim.}

By the claim, for all positions 
$k \in \{1,2,\ldots,n\}$, we have that
$\forall t\geq 0: \val ({H'}_r^t)_k \geq \val ({H'}_r^t)_{i_r}$. Also, by an induction argument identical to
the one we used to argue a similar property for $H'$, we have
$\forall t\geq 0: \val ({H'}_r^t)_{2n} =\val ({H'}_r^t)_{i_r}$.
Thus we may define the game $H''$ from $H'_r$ by applying Lemma \ref{lem:change-child} and changing 
all higher successors of
all coin toss positions $i_j$ to be $i_r$ (instead of $2n$).

For our running example, $H''$ is the game of Figure \ref{fig:w''}. 

\begin{figure}
\centering
\begin{tikzpicture}[scale=0.2,->,>=stealth',shorten >=1pt,auto,node distance=4.2cm*0.5,
                    semithick]
      
\tikzstyle{every state}=[fill=white,draw=black,text=black,font=\small, inner sep=-0.05cm]
\tikzstyle{max}=[state,regular polygon,regular polygon sides=3,inner sep=0cm,minimum height=1.5cm,minimum width=1.5cm]
\tikzstyle{min}=[max,regular polygon rotate=180]

    \node[state,accepting,regular polygon,regular polygon sides=4,inner sep=-0.05cm] (perp)  {GOAL};
    \node[state] (q1) [right of=perp] {1};
    \node[min] (q2) [right of=q1] {2};
    \node[state] (q3) [right of=q2] {3};
    \node[state] (q4) [right of=q3] {4};
    \node[min] (q5) [right of=q4] {5};
    \node[min] (m1) [below of=perp] {6};
    \node[min] (m2) [below of=q1] {7};
    \node[min] (m3) [below of=q2]{8};
    \node[min] (m4) [below of=q3]{9};
    \node[min] (m5) [below of=q4]{10};

\path 
(q1)  edge (perp)
(q1)  edge [out=-45,in=-155]  (q4)
(q2)  edge [bend right] (perp)
(q2)  edge (q1)
(q3)  edge [bend right]  (q1)
(q3)  edge [bend left](q4)
(q4)  edge  (q3)
(q4)  edge  [loop below] (q4)
(q5)  edge  (q4)
(q5)  edge [bend right,in=-110] (q2)
(m1) edge (q1)
(m1) edge (perp)
(m2) edge (q2)
(m2) edge (m1)
(m3) edge (q3)
(m3) edge (m2)
(m4) edge  (q4)
(m4) edge (m3)
(m5) edge  (q5)
(m5) edge (m4);

\end{tikzpicture}
\caption{$H''$, if $H_0$ is the game in Figure \ref{fig:w0}. }
\label{fig:w''}
\end{figure}
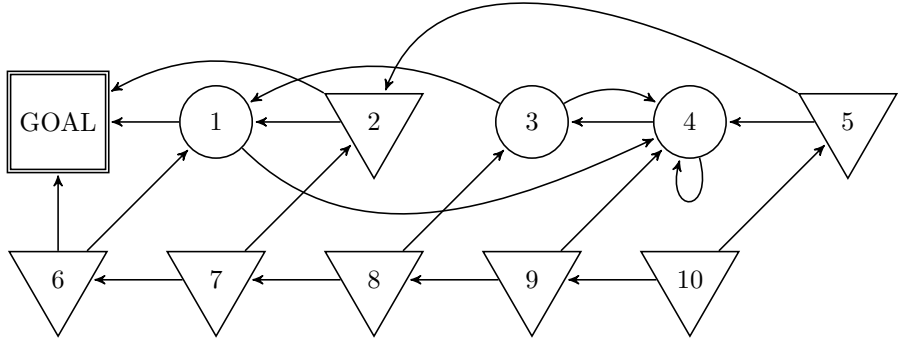

Finally, we arrive at $E_{n,r}$ by 
subsequently removing all ``new'' Min-positions $n+1,\ldots,2n$ and 
changing all successors of all original
Min-positions of $H''$ to be the GOAL position. 

For our running example, $E_{5,3}$ is the extremal game identified by
the proof, and is given in Figure \ref{fig:enr}. Note that the same
game would also be identified for any other game with $r=3$ and $n=5$ (up to the
indexing of the positions). 

It is easy to see that all positions in $E_{n,r}$ have value 1. We have that
each Min-position $k$ in $H''$ satisfies $\val ({H''}^t)_k \geq \val ({H''}^t)_{i_r}$.

We have that $\val (H_n^t)_k \geq \val ({H''}^t)_{i_r}=\val ({E_{n,r}}^t)_{i_r}$ for all $k\in \{1,2,\dots,n\}$, since we found $H''$ either by applying Lemma \ref{lem:change-child} or in a way that did not change the value of any position for any time bound.

Also, $E_{n,r}\in S_{n,r}$, therefore, since at least one possible option for $H_0\in S_{n,r}$ (and therefore $H_n$, since $H_n$ was a reindexing of the positions in $H_0$) was a $t$-extremal game for any $t$, $E_{n,r}$ is extremal. 
This completes the proof of the lemma.

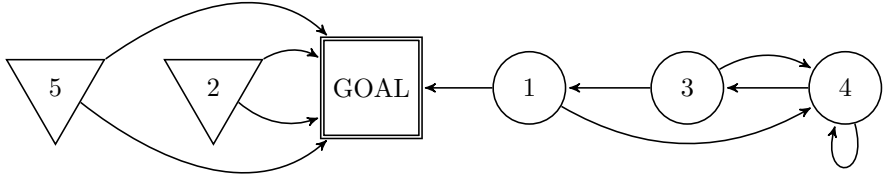
\begin{figure}
\centering
\begin{tikzpicture}[scale=0.2,->,>=stealth',shorten >=1pt,auto,node distance=4.2cm*0.5,
                    semithick]
      
\tikzstyle{every state}=[fill=white,draw=black,text=black,font=\small, inner sep=-0.05cm]
\tikzstyle{max}=[state,regular polygon,regular polygon sides=3,inner sep=0cm,minimum height=1.5cm,minimum width=1.5cm]
\tikzstyle{min}=[max,regular polygon rotate=180]

    \node[state,accepting,regular polygon,regular polygon sides=4,inner sep=-0.05cm] (perp)  {GOAL};
    \node[state] (q1) [right of=perp] {1};
    \node[state] (q2) [right of=q1] {3};
    \node[state] (q3) [right of=q2] {4};

    \node[min] (m1) [left of=perp] {2};
    \node[min] (m2) [left of=m1] {5};

\path 
(m1) edge [bend left] (perp)
(m1) edge [bend right] (perp)
(m2) edge [bend left,in=130] (perp)
(m2) edge [bend right,in=-130] (perp)
(q1) edge (perp)
(q1) edge [bend right] (q3)
(q2) edge (q1)
(q2) edge [bend left] (q3)
(q3) edge (q2)
(q3) edge [loop below] (q3);

\end{tikzpicture}
\caption{The game $E_{5,3}$, upto indexing of the positions.}
\label{fig:enr}
\end{figure}

\end{proof}

Having identified the extremal game $E_{n,r}$, we next estimate
$T(E_{n,r})$.

\begin{lemma}\label{howfast}
For $\forall n, r >0,\epsilon>0,t\geq 2 (\ln 2\epsilon^{-1})  \cdot 2^r, k\in V: \val(E_{n,r})_k - \val(E_{n,r}^T)_k\leq \epsilon$
\end{lemma}
\begin{proof}
We observe that for the purposes of estimating $\val(E_{n,r})_k - \val(E_{n,r}^T)_k$, we can
view $E_{n,r}$ as a game containing $r$ coin toss
positions only, since all Min-positions point directly to GOAL.
Also, when modified value iteration is applied to a game $G$ containing only
coin toss positions, Lemma \ref{lem-sem} implies that 
$\val (G^t)_k$  can be reinterpreted as the probability
that the absorbing Markov process starting in state $k$ is absorbed
within $t$ steps. By the structure of $E_{n,r}$, this is equal to the
probability that a sequence of $t$ fair coin tosses contains
$r$ consecutive tails. This is known to be exactly
$1-F_{t+2}^{(r)}/2^t$, where $F_{t+2}^{(r)}$ is the $(t+2)$'nd Fibonacci
$r$-step number, i.e. the number given by the linear homogeneous recurrence
$F_m^{(r)} = \sum_{i=1}^{r} F_{m-i}^{(r)}$ and the
boundary conditions $F_{m}^{(k)} = 0$, for $m \leq 0$,
$F_{1}^{(r)} = F_{2}^{(r)} = 1$. Asymptotically solving this linear recurrence,
we have that $F_m^{(r)} \leq (\phi_r)^{m-1}$ where $\phi_r$ is the root near 2
to the equation $x + x^{-r} = 2$. Clearly, $\phi_r < 2 - 2^{-r}$,
so \[F_{t+2}^{(r)}/2^t < \frac{(2-2^{-r})^{t+1}}{2^t} = 2(1-2^{-r-1})^{t+1}<2(1-2^{-r-1})^t.\]
Therefore, the
probability that the chain is not absorbed within $t = 2 (\ln 2\epsilon^{-1})  \cdot 2^r$ steps is
at most \[2(1-2^{-r-1})^{2 (\ln 2\epsilon^{-1})  \cdot 2^r} \leq 2e^{-\ln 2\epsilon^{-1}} = \epsilon.\]
\end{proof} 

\begin{corollary}\label{cor:howfast}
For $\forall n, r >0:T(E_{n,r})\leq 2 (\ln 2^{5r+1} ) \cdot 2^r$.
\end{corollary}
\begin{proof}
The proof is by insertion into Lemma \ref{howfast}.
\end{proof}

\section{Conclusions}
We have shown an algorithm solving simple stochastic games obtaining
an improved running time in the worst case compared to previous algorithms, as a function of its
number of coin
toss positions. It is relevant to observe that the {\em best} case
complexity of the algorithm is strongly related to its {\em worst}
case complexity, as the number of iterations of the main loop is fixed
in advance.

As mentioned in the introduction, our paper is partly motivated by a
result of Chatterjee {\em et al.} \cite{C09} analysing the {\em strategy
  iteration} algorithm of the same authors \cite{ChatQest} for the
case of simple stochastic games. We can in fact improve their
analysis, using the techniques of this paper. We combine three facts:
\begin{enumerate}
\item{}(\cite[Lemma 8]{ChatQest}) For a game $G \in S_{n,r}$, after $t$ iterations of the strategy
  iteration algorithm \cite{ChatQest} applied to $G$, the (positional)
  strategy computed for Max guarantees a probability
  of winning of at least $\val(\bar G^t)_k$ against any
  strategy of the opponent when play starts in
  position $k$, 
where $\bar G^t$ is the game defined from $G$ in Definition \ref{def-unmodified value iteration} of this paper.
\item{}For a game $G \in S_{n,r}$ and all $k,t$, $\val(\bar G^{t(n-r+1)})_k
  \geq \val(G^t)_k$. This is a direct consequence of the definitions
  of the two games, and the fact that in an optimal play, either a coin
  toss position is encountered at least after every $n-r+1$ moves of the
  pebble, or never again.
\item{}Corollary \ref{cor:howfast} of the present paper.
\end{enumerate}
These three facts together implies that the strategy iteration
algorithm after $2 (\ln 2^{5r+1})  \cdot 2^r(n-r+1)$ iterations has computed a strategy that
guarantees the values of the game within an additive error of
$2^{-5r}$, for $r\geq 6$. As observed by Chatterjee {\em et al.} \cite{C09}, such a
strategy is in fact optimal. Hence, we conclude that their strategy
iteration algorithm terminates in time $2^r n^{O(1)}$. This improves their analysis of
the algorithm significantly, but still yields a bound on its worst
case running time inferior to the worst case running time of the
algorithm presented here. On the other hand, unlike the algorithm
presented in this paper, their algorithm has the desirable property that it may terminate faster than its worst case analysis suggests.
\bibliographystyle{plain}
\bibliography{gurvich}
\end{document}